\documentclass[12pt]{amsart}
\usepackage{amssymb}
\usepackage{color}
\usepackage{amsmath,epic,curves,amscd}
\usepackage[english]{babel}
\usepackage{graphicx}
\usepackage{comment}
\usepackage{appendix}
\usepackage{mathdots}
\usepackage[all]{xy}
\newtheorem{claim}{}[section]
\newtheorem{theorem}[claim]{Theorem}

\newtheorem{lemma}[claim]{Lemma}
\newtheorem{proposition}[claim]{Proposition}

\theoremstyle{remark}

\renewenvironment{proof}{\noindent{\it Proof. \hskip0pt}}
                      {$\square$\par\medskip}

\begin{document}
\baselineskip 6.0 truemm
\parindent 1.5 true pc

\newcommand\lan{\langle}
\newcommand\ran{\rangle}
\newcommand\tr{{\text{\rm Tr}}\,}
\newcommand\ot{\otimes}
\newcommand\ol{\overline}
\newcommand\join{\vee}
\newcommand\meet{\wedge}
\renewcommand\ker{{\text{\rm Ker}}\,}
\newcommand\image{{\text{\rm Im}}\,}
\newcommand\id{{\text{\rm id}}}
\newcommand\tp{{\text{\rm tp}}}
\newcommand\pr{\prime}
\newcommand\e{\epsilon}
\newcommand\la{\lambda}
\newcommand\inte{{\text{\rm int}}\,}
\newcommand\ttt{{\text{\rm t}}}
\newcommand\spa{{\text{\rm span}}\,}
\newcommand\conv{{\text{\rm conv}}\,}
\newcommand\rank{\ {\text{\rm rank of}}\ }
\newcommand\re{{\text{\rm Re}}\,}
\newcommand\ppt{\mathbb T}
\newcommand\rk{{\text{\rm rank}}\,}
\newcommand\SN{{\text{\rm SN}}\,}
\newcommand\SR{{\text{\rm SR}}\,}
\newcommand\HA{{\mathcal H}_A}
\newcommand\HB{{\mathcal H}_B}
\newcommand\HC{{\mathcal H}_C}
\newcommand\CI{{\mathcal I}}
\newcommand{\bra}[1]{\langle{#1}|}
\newcommand{\ket}[1]{|{#1}\rangle}
\newcommand\cl{\mathcal}
\newcommand\idd{{\text{\rm id}}}
\newcommand\OMAX{{\text{\rm OMAX}}}
\newcommand\OMIN{{\text{\rm OMIN}}}
\newcommand\diag{{\text{\rm Diag}}\,}
\newcommand\calI{{\mathcal I}}
\newcommand\bfi{{\bf i}}
\newcommand\bfj{{\bf j}}
\newcommand\bfk{{\bf k}}
\newcommand\bfl{{\bf l}}
\newcommand\bfp{{\bf p}}
\newcommand\bfq{{\bf q}}
\newcommand\bfzero{{\bf 0}}
\newcommand\bfone{{\bf 1}}
\newcommand\im{{\mathcal R}}
\newcommand\ha{{\frac 12}}
\newcommand\xx{{\text{\sf X}}}
\newcommand\sa{{\text{\rm sa}}}
\newcommand\Pz{P^{\rm z}}
\newcommand\Pn{P^{\rm n}}

\title{Three qubit separable states of length ten\\
with unique decompositions}

\author{Seung-Hyeok Kye}
\address{Department of Mathematics and Institute of Mathematics, Seoul National University, Seoul 151-742, Korea}
\email{kye at snu.ac.kr}
\thanks{partially supported by NRF-2017R1A2B4006655.}



\begin{abstract}
We construct one parameter families of three qubit separable states with length ten, which is strictly greater than the whole
dimension eight. These states are located on the boundary of the convex set of all separable states, but they are in the interior
of the convex set of all states with positive partial transposes. They are also decomposed into the convex sum of ten
pure product states in a unique way.
\end{abstract}

\maketitle

\section{Introduction}

Entanglement is one of the key notions in the current quantum
information theory, and it is an important research topic to
distinguish entanglement from separability. Recall that a
multi-partite state is said to be separable if it is a convex
combination of pure product states, and entangled if it is not
separable. The length of separable states is one of the notions to
understand the convex structures of the convex set $\mathcal S$
consisting of all separable states, as it was considered by several
authors
\cite{{DiVin},{p-horo},{lockhart},{sanpera-t-g},{sko_10},{uhlmann},{wootters}}
in the early stages of quantum information theory.

The {\sl length} $\ell(\varrho)$ of a separable state $\varrho$ is
defined by the smallest number of pure product states with which
$\varrho$ may be expressed as a convex combination. Physically, it
represents the minimal physical effort to implement the state. It is
clear that the length $\ell(\varrho)$ is bigger than or equal to the
rank of $\varrho$. It was known in \cite{DiVin} that the length may
be strictly bigger than the rank. On the other hand, it was shown
\cite{chen_dj_2xd} for $2\otimes 3$ bi-partite case that the length
$\ell(\varrho)$ of a separable state $\varrho$ coincides with the
maximum of the ranks of $\varrho$ and its partial transpose
$\varrho^\Gamma$. This is not the case for the $3\otimes 3$ system,
by examples \cite{ha-kye-sep-face} of separable states $\varrho$
with $\ell(\varrho)=6$ but $\rk \varrho=\rk\varrho^\Gamma=5$.

The length cannot exceed the whole affine dimension of the system by the classical
Carat\' eodory theorem, as it was observed in \cite{p-horo}. One natural question is to find the possible maximum of lengths.
It was known \cite{sanpera-t-g} that $\ell(\varrho)\le 4$ for every $2\otimes 2$ separable state $\varrho$. The above mentioned result
\cite{chen_dj_2xd} on $2\otimes 3$ separable states tells us $\ell(\varrho)\le 6$ for those separable states.
On the other hand, it was shown in \cite{chen_dj_semialg} that the length of an $m\otimes n$ state may exceeds $mn$
whenever $(m-2)(n-2)>1$, or equivalently $\max\{m,n\}\ge 4$ and $\min\{m,n\}\ge 3$.
In the $3\otimes 3$ and $2\otimes 4$ systems, examples of separable states of length $10$ were found in
\cite{ha_kye_exposedness_length,ha-kye-sep-face} and \cite{ha-kye-sep_2x4}, respectively, and examples
in \cite{ha-kye-sep-face} have been analyzed in \cite{chen_dj_boundary}.
We refer to \cite{{chen_dj_maximal face},chen_dj_length_fil} for further results on the related topics.

In this paper, we consider the three qubit system, the simplest multi-partite case. A result in
\cite{chen_dj_semialg} tells us that there must exist a three qubit separable state of length ten.
The main purpose of this paper is to construct explicit examples of such states.
Lengths of some three qubit separable states have been calculated in \cite{han_kye_phase} but do not exceed eight, even though some of them
exceed the maximum ranks  of partial transposes. The states we constructed turn out to be boundary separable states
of full ranks in the sense of \cite{chen_dj_boundary}, that is, it is on the boundary of the convex set $\mathcal S$,
and all the partial transposes have full ranks. Recall that all the partial transposes of a state have full ranks if and only if it is
in the interior of the convex set $\mathcal T$ of all PPT states.

The main idea is to use the duality
\cite{kye_3qb_EW} between $n$-partite separable states
and positive multi-linear maps.
For our purpose, we consider the positive bi-linear maps in $2\times 2$ matrices
which have been constructed in \cite{kye_3qb_EW,kye_3q_exposed}, and mix variants of them
to get positive bi-linear maps whose dual faces have exactly ten pure product states.
The positive maps constructed have the full spanning property, and so, the interior of the dual faces
are located in the interior of the convex cone of all PPT states. See \cite{kye_3qb_EW}.
Furthermore, the faces are affinely isomorphic
to the $9$-dimensional simplex with ten extreme points, and so any interior points of the faces are separable states
of length ten and are decomposed into the sum of pure product states in a unique way.

Separable states with unique decompositions also have been studied by several authors
\cite{{alfsen},{alfsen_2},{chen_dj_ext_PPT},{cohen},{ha-kye-sep_2x4},{ha-kye-sep-face},{kirk}}.
The results tell us that separable states generically have unique decompositions if the lengths are sufficiently small.
In the multi-partite case of $d_1\ot \cdots\ot d_n$ system, it was shown in \cite{ha-kye-multi-unique} that
generic choices of $k$ product vectors give rise to separable states with unique decompositions whenever $k\le \sum_{i=1}^n (d_i-1)$,
by taking convex sums of corresponding pure product states. In the three qubit case,
generic choices of four product vectors still give rise to separable states with unique decompositions.
We know \cite{juhan,walgate} that generic five dimensional subspaces of $\mathbb C^2\ot\mathbb C^2\ot\mathbb C^2$  have
six product vectors. It was also shown \cite{ha-kye-multi-unique}
that they make three qubit separable states of length six with unique decompositions.
Our construction gives examples of three qubit separable states of length ten with unique decompositions.

We collect in the next section basic material about the map in
\cite{kye_3qb_EW}, and construct in Section 3 the positive bi-linear
map we are looking for.

The author is grateful to Kyung Hoon Han for valuable comments and discussion.

\section{Boundary separable states with full ranks}

For a given multi-linear map $\phi$ from $M_{d_1}\times \cdots\times M_{d_{n-1}}$ into $M_{d_n}$,
we associate the Choi matrix $W_\phi$ in $M_{d_1}\otimes \cdots\otimes M_{d_{n-1}}\otimes M_{d_n}$ by
$$
W_\phi =\sum_{i_1,j_1,\dots, i_{n-1},j_{n-1}}
|i_1\ran\lan j_1|\otimes \cdots\otimes |i_{n-1}\ran\lan j_{n-1}|\ot \phi(|i_1\ran\lan j_1|,\cdots, |i_{n-1}\ran\lan j_{n-1}|),\\
$$
and define the bi-linear pairing
$$
\lan \varrho,\phi\ran = {\text{\rm Tr}}(\varrho W^\ttt_\phi)
$$
for an $(n-1)$-linear map $\phi$ and an $n$-partite state $\varrho$, where $W^\ttt_\phi$ denotes the transpose
of $W_\phi$. We recall that an $(n-1)$-linear map is positive if $\phi(x_1,\dots,x_{n-1})$ is positive (semi-definite)
whenever all of $x_1,\dots, x_{n-1}$ are positive.
It was shown in \cite{kye_3qb_EW} that $\varrho$ is separable if and only if
$\lan\varrho,\phi\ran\ge 0$ for every positive multi-linear map $\phi$. This tells us that every entanglement witness
must be the Choi matrix $W_\phi$ of a positive multi-linear map $\phi$.

We have constructed in \cite{kye_3qb_EW} the positive bilinear map
$\phi:M_2\times M_2\to M_2$ which sends $([x_{ij}], [y_{ij}])\in
M_2\times M_2$ to
$$
\left(\begin{matrix}
sx_{22}y_{11}&
x_{12}y_{12}-x_{12}y_{21}+x_{21}y_{12}+x_{21}y_{21}\\
x_{12}y_{12}+x_{12}y_{21}-x_{21}y_{12}+x_{21}y_{21}
&t x_{11}y_{22}
\end{matrix}\right)\in M_2,
$$
where $s,t$ are positive numbers with $st=8$. See also
\cite{kye_3q_exposed} for the motivation of construction. We use the
parameter $u>0$ to express the the corresponding entanglement
witness, or equivalently the corresponding Choi matrix:
$$
W=\left(\begin{matrix}
\cdot&\cdot&\cdot&\cdot&\cdot&\cdot&\cdot &1\\
\cdot&\cdot&\cdot&\cdot&\cdot&\cdot&1&\cdot\\
\cdot&\cdot&\cdot&\cdot&\cdot&-1&\cdot&\cdot\\
\cdot&\cdot&\cdot&\sqrt 8 u^{-1}&1&\cdot&\cdot&\cdot\\
\cdot&\cdot&\cdot&1&\sqrt 8 u&\cdot&\cdot&\cdot\\
\cdot&\cdot&-1&\cdot&\cdot&\cdot&\cdot&\cdot\\
\cdot&1&\cdot&\cdot&\cdot&\cdot&\cdot&\cdot\\
1&\cdot&\cdot&\cdot&\cdot&\cdot&\cdot&\cdot
\end{matrix}\right),
$$
with respect to the lexicographic order $000$, $001$, $010$, $011$, $100$,
$101$, $110$, $111$. It was shown \cite{kye_3qb_EW} that the map
$\phi$ is indecomposable positive bi-linear map. In fact, $\phi$ has
the full spanning property which implies indecomposability. See also
\cite{han_kye_genuine} for indecomposability. More recently, the
author showed \cite{kye_3q_exposed} that the map $\phi$ generates an
exposed extreme ray of the cone of all positive bi-linear maps between
$2\times 2$ matrix algebras.

We note that all the entries of $W$ are zero except for diagonal and anti-diagonals. Such matrices
are called {\sf X}-shaped, and are of the form
$$
X(a,b,c)= \left(
\begin{matrix}
a_1 &&&&&&& c_1\\
& a_2 &&&&& c_2 & \\
&& a_3 &&& c_3 &&\\
&&& a_4&c_4 &&&\\
&&& \bar c_4& b_4&&&\\
&& \bar c_3 &&& b_3 &&\\
& \bar c_2 &&&&& b_2 &\\
\bar c_1 &&&&&&& b_1
\end{matrix}
\right),
$$
for vectors
$a,b\in\mathbb R^4$ and $c\in\mathbb C^4$. Many important three qubit states arise in this form. For example,
Greenberger-Horne-Zeilinger diagonal states are {\sf X}-shaped, and $\varrho=X(a,b,c)$ is GHZ diagonal
if and only if $a=b$ and $c\in\mathbb R^4$. With this notation, we note that the above entanglement witnesses $W$ is given by
$$
W=X((0,0,0,\sqrt 8u^{-1}), (0,0,0,\sqrt 8 u),(1,1,-1,1)).
$$

We consider the dual face $W^\prime$ of $W$, which consists of all separable states $\varrho$ such that $\lan W,\varrho\ran=0$.
In order to understand the structures of the dual face, it is desirable to find all the extreme points of the face,
that is, pure product states $\varrho=|\xi\ran\lan\xi|$ satisfying the relation $\lan W,\varrho\ran=0$.
We have found in \cite{kye_3q_exposed} all such product vectors $|\xi\ran$'s.
Note that such a product vector $|\xi\ran$ satisfies the relation
$$
\lan\bar\xi |W|\bar\xi\ran =\lan W,|\xi\ran\lan\xi|\ran = 0.
$$
In this case, we say that the witness $W$ {\sl kills} the product vector $|\xi\ran$, and denote by $P_W$ the set of
all product vectors killed by $W$:
$$
P_W:=\{|\xi\ran=|x\ran\ot |y\ran\ot |z\ran: \lan W,|\xi\ran\lan\xi|\ran =0\}.
$$
Our discussion tells us that there is a one-to-one correspondence
between $P_W$ and the set of all extreme points of the dual face $W^\prime$ of $W$.

We denote by $\Pz_W$ the set of all product vectors in $P_W$ which has zero entries, and by
$P_W^{\rm n}$ the complement $P_W\setminus P_W^{\rm z}$.
For $i,j=0,1$, we denote by $E_{\square ij}$ the set of all three qubit product vectors of the form $|x\ran\ot |i\ran \ot |j\ran$
with $|x\ran\in\mathbb C^2$. We also define $E_{i\square j}$ and $E_{ij\square}$ similarly:
$$
\begin{aligned}
E_{\square ij}&=\{|x\ran \otimes |i\ran \otimes |j\ran: |x\ran\in\mathbb C^2\},\\
E_{i\square j}&=\{|i\ran \otimes |x\ran \otimes |j\ran: |x\ran\in\mathbb C^2\},\\
E_{ij\square }&=\{|i\ran \otimes |j\ran \otimes |x\ran: |x\ran\in\mathbb C^2\}.
\end{aligned}
$$
Note that all the above sets are parameterized by the two dimensional sphere, up to scalar multiplications. It is easy to see that
the set $\Pz_W$ of product vectors in $P_W$ with zero entries is given by
\begin{equation}\label{pzw}
\Pz_W=E_{\square 01}\cup E_{\square 10}\cup E_{0\square 0}\cup E_{1\square 1}\cup E_{00\square}\cup E_{11\square}.
\end{equation}

We proceed to look for product vectors belonging to $P_W^{\text {\rm n}}$.
For a given triplet ${\bf p}=(p,q,r)\in\mathbb R_+^3$ of positive numbers and $\Lambda=(\alpha,\beta,\gamma)\in\mathbb T^3$, we define
the three qubit product vector $|\eta({\bf p},\Lambda)\ran$ by
$$
|\eta({\bf p},\Lambda)\ran=(pqr)^{-\frac 12}
(p,\alpha)^\ttt\otimes (q,\beta)^\ttt\otimes (r,\gamma)^\ttt\in\mathbb C^2\otimes\mathbb C^2\otimes \mathbb C^2.
$$
Denote by $\omega=e^{{\rm i}\frac \pi 4}$ the eighth root of unity, and take eight $\Lambda$'s in $\mathbb T^3$ as follows:
$$
\begin{aligned}
\Lambda_1&=(+\omega^3,+\omega^1,+\omega^7),\\
\Lambda_2&=(+\omega^3,-\omega^1,-\omega^7),\\
\Lambda_3&=(-\omega^3,+\omega^1,-\omega^7),\\
\Lambda_4&=(-\omega^3,-\omega^1,+\omega^7),\\
\Lambda_j&=\bar\Lambda_{j-4},\quad j=5,6,7,8.
\end{aligned}
$$
It was shown in \cite{kye_3q_exposed} that a product vector without zero entries belongs to $P_W$
if and only if it is of the form
$$
|\eta_j({\bf p})\ran=:|\eta({\bf p},\Lambda_j)\ran
$$
up to scalar multiplication for some $j=1,\dots,8$ and ${\bf p}=(p,q,r)\in\mathbb R_+^3$ satisfying the relation $pq^{-1}r^{-1}=u$,
or equivalently
$$
pq^{-1}=ur,\qquad
qr=pu^{-1}\quad {\text{\rm or}}\quad
rp^{-1}=q^{-1}u^{-1}.
$$
For a given subset $\Sigma\subset \mathbb R^3_+$, we also define the set $F_\Sigma$ of product vectors by
$$
F_{\Sigma}= \{|\eta_j({\bf p})\ran: {\bf p}\in \Sigma, \ j=1,2,\dots,8\}.
$$
We have seen that the set $\Pn_W$ of product vectors in $P_W$ without zero entries is given by
$F_S$ with
$$
S=\{(p,q,r)\in\mathbb R_+^3: pq^{-1}r^{-1}=u\}.
$$
In short, we see that $P_W$ consists of product vectors in (\ref{pzw}) and those in $F_S$.

We fix a point ${\bf p}$ in the surface $S$, and
write $\varrho_j=:|\eta_j({\bf p})\ran  \lan \eta_j({\bf p})|$ for $j=1,2,\dots,8$.
By the argument in Section III of \cite{han_kye_phase}, we see that
the pure product states $\varrho_1, \varrho_2,\varrho_3$ and $\varrho_4$ share the common {\sf X}-part,
which is a separable {\sf X}-state of rank four. Furthermore, this {\sf X}-state is uniquely decomposed into the average of
$\varrho_1, \varrho_2,\varrho_3$ and $\varrho_4$. See \cite{ha_han_kye_multi_qubit} for multi-qubit analogue.
By a direct calculation, we have
$$
\varrho_9
:=\frac 14(\varrho_1+\varrho_2+\varrho_3+\varrho_4)
=X(a_{\bfp},b_\bfp, (\omega^5, \omega^3,\omega^7,\omega^5)),
$$
with the notations
$$
a_\bfp= (p^2u^{-1}, q^2u, r^2u, u),
\qquad
b_\bfp=(p^{-2}u, q^{-2}u^{-1}, r^{-2}u^{-1}, u^{-1}).
$$
On the other hand, the states $\varrho_5,\varrho_6,\varrho_7$ and $\varrho_8$ share the {\sf X}-part
$$
\varrho_{10}
:=\frac 14(\varrho_5+\varrho_6+\varrho_7+\varrho_8)
=X(a_\bfp, b_\bfp, (\omega^3, \omega^5,\omega,\omega^3))
$$
which is again a rank four separable state with unique decomposition.
Note that the average of all of them is given by
\begin{equation}\label{varrho_p}
\varrho_{\bf p}
:=\frac 12(\varrho_9+\varrho_{10})=\frac 18\sum_{j=1}^8\varrho_j
=X(a_\bfp, b_\bfp, \textstyle{\frac 1{\sqrt2}}(-1, -1,+1,-1)).
\end{equation}
Comparing the diagonal and anti-diagonal entries, we see that $\varrho_{\bf p}$ has the full rank eight.
This shows that the eight product vectors $|\eta_j({\bf p})\ran$ with $k=1,2,\dots.8$ make a basis of
$\mathbb C^2\ot\mathbb C^2\ot\mathbb C^2$ for each fixed ${\bf p}\in S$.

We denote by $\Gamma_A$, $\Gamma_B$ ad $\Gamma_C$ the partial transposes
with respect to the $A$, $B$ and $C$ parties, respectively. Then we have
$$
\begin{aligned}
X(a,b,c)^{\Gamma_A}&=X(a,b,(\bar c_4,\bar c_3,\bar c_2,\bar c_1)),\\
X(a,b,c)^{\Gamma_B}&=X(a,b,(c_3,c_4,c_1,c_2),\\
X(a,b,c)^{\Gamma_C}&=X(a,b,(c_2,c_1,c_4,c_3).
\end{aligned}
$$
Therefore, we see that all the partial transposes
$$
\begin{aligned}
(\varrho_{\bf p})^{\Gamma_A}
&=X(a_\bfp, b_\bfp, \textstyle{\frac 1{\sqrt2}}(-1, +1,-1,-1)),\\
(\varrho_{\bf p})^{\Gamma_B}
&=X(a_\bfp, b_\bfp,  \textstyle{\frac 1{\sqrt2}}(+1, -1,-1,-1)),\\
(\varrho_{\bf p})^{\Gamma_C}
&=X(a_\bfp, b_\bfp,  \textstyle{\frac 1{\sqrt2}}(-1, -1,-1,+1))
\end{aligned}
$$
also have the full ranks. Hence, we have the following:

\begin{proposition}\label{full-ranks}
The state $\varrho_{\bf p}$ is a boundary separable state with full ranks
for every ${\bf p}$ on the surface $S$.
Especially, eight product vectors $\{|\eta_j({\bf p})^\gamma\ran:j=1,2,\dots,8\}$ form a basis for each
partial conjugate operation $\gamma$ and $\bfp\in S$.
\end{proposition}

\section{Construction}

We are going to construct entanglement witnesses whose dual faces have only finitely many extreme points.
The main idea is to apply the bit-flip operators and partial transposes to the witness
$$
W=X((0,0,0,\sqrt 8u^{-1}), (0,0,0,\sqrt 8 u),(1,1,-1,1)),
$$
to get witnesses
with the same anti-diagonals as $W$ but different diagonals from $W$.
We will mix some of them to get entanglement witnesses whose dual faces have exactly ten extreme points.
Recall the relation
$$
(W_1+W_2)^\prime=W_1^\prime\cap W_2^\prime
$$
for entanglement witnesses $W_1$ and $W_2$.

We denote by $\sigma$ the bit-flip operator
$\left(\begin{matrix}0&1\\1&0\end{matrix}\right)$ on $\mathbb C^2$, and define
the operators on $\mathbb C^2\ot\mathbb C^2\ot\mathbb C^2$ by
$$
\sigma_A=\sigma\ot I\ot I,\qquad \sigma_B=I\ot \sigma\ot I,\qquad \sigma_C=I\ot I\ot\sigma,
$$
where $I$ denotes the identity operator. We note that
$$
\begin{aligned}
\sigma_A X(a,b,c)\sigma_A&=X((b_4,b_3,b_2,b_1),(a_4,a_3,a_2,a_1),(\bar c_4,\bar c_3,\bar c_2,\bar c_1)),  \\
\sigma_B X(a,b,c)\sigma_B&=X((a_3,a_4,a_1,a_2),(b_3,b_4,b_1,b_2),(c_3,c_4,c_1,c_2)),  \\
\sigma_C X(a,b,c)\sigma_C&=X((a_2,a_1,a_4,a_3),(b_2,b_1,b_4,b_3),(c_2,c_1,c_4,c_3)).
\end{aligned}
$$
Therefore, we have
$$
\begin{aligned}
W_A&:=\sigma_A W^{\Gamma_A}\sigma_A=X((\sqrt 8u,0,0,0), (\sqrt 8 u^{-1},0,0,0),(1,1,-1,1)),\\
W_B&:=\sigma_B W^{\Gamma_B}\sigma_B=X((0,\sqrt 8u^{-1},0,0), (0,\sqrt 8 u,0,0),(1,1,-1,1)),\\
W_C&:=\sigma_C W^{\Gamma_C}\sigma_C=X((0,0,\sqrt 8u^{-1},0), (0,0,\sqrt 8 u,0),(1,1,-1,1)).
\end{aligned}
$$

Now, we are looking for product vectors $|\xi\ran=|x\ran\ot|y\ran\ot|z\ran$
which are killed by the witness $W_A$. We note that
$$
\begin{aligned}
\lan W_A, |\xi\ran\lan\xi|\ran
&=\lan \sigma_A W^{\Gamma_A}\sigma_A,|\xi\ran\lan\xi|\ran\\
&=\lan W^{\Gamma_A},|\sigma_A\xi\ran\lan\sigma_A\xi|\ran\\
&=\lan W^{\Gamma_A}, |\sigma x\ran| y\ran| z\ran\lan\sigma x|\lan y|\lan z|\ran\\
&=\lan W, |\sigma\bar x\ran| y\ran| z\ran\lan\sigma\bar x|\lan y|\lan z|\ran.
\end{aligned}
$$
Motivated by this relation, we write
$$
|\xi_A\ran=  |\sigma\bar x\ran\ot| y\ran\ot| z\ran,\quad
|\xi_B\ran=  | x\ran\ot|\sigma\bar y\ran\ot| z\ran,\quad
|\xi_C\ran=  | x\ran\ot| y\ran\ot|\sigma\bar z\ran,
$$
for a product vector $|\xi\ran=|x\ran\ot|y\ran\ot|z\ran$.
Then we see that a product vector $|\xi\ran$ satisfies the relation $|\xi\ran\in P_{W_A}$ if and only if
$|\xi_A\ran\in P_W$, and similarly for $B$ and $C$ parties.
Therefore, we see by (\ref{pzw}) that the set $\Pz_{W_A}$, $\Pz_{W_B}$ and $\Pz_{W_C}$ are given by
$$
\begin{aligned}
\Pz_{W_A}:=E_{\square 01}\cup E_{\square 10}\cup E_{1\square 0}\cup E_{0\square 1}\cup E_{10\square}\cup E_{01\square},\\
\Pz_{W_B}:=E_{\square 11}\cup E_{\square 00}\cup E_{0\square 0}\cup E_{1\square 1}\cup E_{01\square}\cup E_{10\square},\\
\Pz_{W_C}:=E_{\square 00}\cup E_{\square 11}\cup E_{0\square 1}\cup E_{1\square 0}\cup E_{00\square}\cup E_{11\square}.
\end{aligned}
$$

Next, we look for product vectors in $P_{W_A}$ without zero entry.
To do this, we first note the relation
$$
\sigma(p,\bar\alpha)^\ttt=(\bar\alpha,p)^\ttt=p\bar\alpha (p^{-1},\alpha)^\ttt,
$$
from which we have
$$
|\eta_j({\bf p})_A\ran=|\eta_j(p^{-1},q,r)\ran,\qquad j=1,2,\dots,8,
$$
up to scalar multiplications, and similar relations for $|\eta_j({\bf p})_B\ran$ and $|\eta_j({\bf p})_C\ran$;
$$
\begin{aligned}
|\eta_j({\bf p})_B\ran&=|\eta_j(p,q^{-1},r)\ran,\qquad j=1,2,\dots,8,\\
|\eta_j({\bf p})_C\ran&=|\eta_j(p,q,r^{-1})\ran,\qquad j=1,2,\dots,8.
\end{aligned}
$$
Now, we define surfaces
$$
\begin{aligned}
S_A&=\{(p,q,r): (p^{-1},q,r)\in S\}=\{(p,q,r)\in\mathbb R^3_+: p^{-1}q^{-1}r^{-1}=u\},\\\
S_B&=\{(p,q,r): (p,q^{-1},r)\in S\}=\{(p,q,r)\in\mathbb R^3_+: pqr^{-1}=u\},\\
S_C&=\{(p,q,r): (p,q,r^{-1})\in S\}=\{(p,q,r)\in\mathbb R^3_+: pq^{-1}r=u\}.
\end{aligned}
$$
Then we have
$$
\Pn_{W_A}=F_{S_A},\qquad \Pn_{W_B}=F_{S_B},\qquad  \Pn_{W_C}=F_{S_C}.
$$

Now, we consider the intersections of two sets to get
$$
\begin{aligned}
P_W\cap P_{W_A}&=E_{\square 01}\cup E_{\square 10}\cup F_{S\cap S_A},\\
P_W\cap P_{W_B}&=E_{0\square 0}\cup E_{1\square 1}\cup F_{S\cap S_B},\\
P_W\cap P_{W_C}&=E_{00\square }\cup E_{11\square }\cup F_{S\cap S_C}.
\end{aligned}
$$
We also have the following:
$$
\begin{aligned}
P_{W_A}\cap P_{W_B}&=E_{01\square}\cup E_{10\square}\cup F_{S_A\cap S_B},\\
P_{W_B}\cap P_{W_C}&=E_{\square 00}\cup E_{\square 11}\cup F_{S_B\cap S_C},\\
P_{W_C}\cap P_{W_A}&=E_{0\square 1}\cup E_{1\square 0}\cup F_{S_C\cap S_A}.
\end{aligned}
$$
We note that the sets $P_W$, $P_{W_A}$, $P_{W_B}$ and $P_{W_C}$ are parameterized by six spheres and eight two-dimensional
surfaces, respectively. On the other hand, an intersection of two of them is parameterized by two spheres and eight curves.
To be more precise, we note that
$$
S\cap S_A=\{(1,q,q^{-1}u^{-1}): 0<q<\infty\},
$$
and so we see that sides of the eight curves $F_{S\cap S_A}$ approach to the sphere given by $E_{\square 01}$
and the other sides approach to
the sphere given by $E_{\square 10}$. The same description works for other intersections.

Now, we consider the intersection of three of them. For example, we have
$$
\begin{aligned}
P_W\cap P_{W_A}\cap P_{W_B}
&=(P_W\cap P_{W_A})\cap(P_W\cap P_{W_B})\\
&=[(E_{\square 01}\cup E_{\square 10})\cap (E_{0\square 0}\cup E_{1\square 1})]\cup F_{S\cap S_A\cap S_B}.
\end{aligned}
$$
We see that $(E_{\square 01}\cup E_{\square 10})\cap (E_{0\square 0}\cup E_{1\square 1})$ consists of two product vectors
$|1\ran\ot |0\ran \ot |1\ran$ and $|0\ran\ot |1\ran \ot |0\ran$. We also see that
$(p,q,r)\in S\cap S_A\cap S_B$ if and only if the relation
$$
pq^{-1}r^{-1}=u,\qquad p^{-1}q^{-1}r^{-1}=u,\qquad pqr^{-1}=u
$$
holds if and only if $\bfp=(1,1,u^{-1})$. Therefore, $F_{S\cap S_A\cap S_B}$ consists of eight product vectors. In short,
the set $P_W\cap P_{W_A}\cap P_{W_B}$ consists of the following ten product vectors:
\begin{equation}\label{ten}
P_W\cap P_{W_A}\cap P_{W_B}=\{|010\ran, |101\ran, |\eta_j(1,1,u^{-1})\ran\ (j=1,2,\dots,8)\}.
\end{equation}
The state $\varrho_\bfp$ in (\ref{varrho_p}) with $\bfp=(1,1,u^{-1})$
is given by
$$
\varrho_\bfp=X((u^{-1},u,u^{-1},u), (u,u^{-1},u,u^{-1}),\textstyle{\frac 1{\sqrt2}}(-1, -1,+1,-1)),
$$
and it is easily checked that $\lan W+W_A+W_B, \varrho_\bfp\ran=0$.

By the exactly same way, we have
$$
\begin{aligned}
P_W\cap P_{W_B}\cap P_{W_C}&=\{|000\ran, |111\ran, |\eta_j(u,1,1)\ran\ (j=1,2,\dots,8)\},\\
P_W\cap P_{W_C}\cap P_{W_A}&=\{|001\ran, |110\ran, |\eta_j(1,u^{-1},1)\ran\ (j=1,2,\dots,8)\},\\
P_{W_A}\cap P_{W_B}\cap P_{W_C}&=\{|011\ran, |100\ran, |\eta_j(u,u^{-1},u^{-1})\ran\ (j=1,2,\dots,8)\}.
\end{aligned}
$$
We also have
$$
\begin{aligned}
\varrho_{(u,1,1)}&=X((u,u,u,u), (u^{-1},u^{-1},u^{-1},u^{-1}),\textstyle{\frac 1{\sqrt2}}(-1, -1,+1,-1)),\\
\varrho_{(1,u^{-1},1)}&=X((u^{-1},u^{-1},u,u), (u,u,u^{-1},u^{-1}),\textstyle{\frac 1{\sqrt2}}(-1, -1,+1,-1)),\\
\varrho_{(u,u^{-1},u^{-1})}&=X((u,u^{-1},u^{-1},u), (u^{-1},u,u,u^{-1}),\textstyle{\frac 1{\sqrt2}}(-1, -1,+1,-1)).
\end{aligned}
$$

So far, we have seen that the dual face
$$
(W+W_A+W_B)^\prime=W^\prime\cap W_A^\prime\cap W_B^\prime
$$
of the witness $W+W_A+W_B$ has exactly ten pure product states.
So, this face is the convex hull of the ten extreme points. It is well known that the convex hull
of finitely many points on an affine manifold in a real vector space is a simplex
if and only if they are linearly independent.

\begin{theorem}\label{main}
Suppose that $F$ is the dual face of one of the entanglement witnesses
$$
W+W_A+W_B,\quad W+W_B+W_C,\quad W+W_C+W_A, \quad W_A+W_B+W_C.
$$
Then we have the following:
\begin{enumerate}
\item[(i)]
$F$ is affinely isomorphic to the $9$-dimensional simplex with ten extreme points,
\item[(ii)]
every point of $F$ is a separable state with unique decomposition into the sum of pure product states,
\item[(iii)]
every interior point of $F$ is a boundary separable state with full ranks, and has length ten.
\end{enumerate}
\end{theorem}

\begin{proof}
We will prove for the entanglement $W+W_A+W_B$ only.
For the statement (i), it remains to show that the ten pure product states arising from
ten product vectors in (\ref{ten}) are linearly independent in the real vector space of all three qubit self-adjoint
matrices. Write $|\eta_j\ran=|\eta_j(1,1,u^{-1})\ran$ for $j=1,2,\dots,8$, and suppose that
$$
a|010\ran\lan 010|+b|101\ran\lan 101|+\sum_{j=1}^8c_j|\eta_j\ran\lan\eta_j|=0,
$$
with real numbers $a,b$ and $c_j$. Then we have
$$
0=a|010\ran\lan 010|000\ran+b|101\ran\lan 101|000\ran+\sum_{j=1}^8c_j|\eta_j\ran\lan\eta_j|000\ran
=\sum_{j=1}^8c_j|\eta_j\ran\lan\eta_j|000\ran.
$$
We have already seen that eight vectors $\{|\eta_j\ran:j=1,\dots,8\}$ form a basis, and so
$c_j\lan\eta_j|000\ran=0$ for each $j=1,2,\dots,8$. Since $|\eta_j\ran$ has no zero entry, we conclude that
$\lan\eta_j|000\ran\neq 0$ and $c_j=0$ for each $j=1,2,\dots,8$.
The second statement (ii) follows from (i).
For (iii), we note that every interior point of $F$ must be expressed by
$$
a|010\ran\lan010|+ b|101\ran\lan 101| +\sum_{j=1}^8 c_j |\eta_j\ran\lan\eta_j|
$$
 in a unique way with $a>0$, $b>0$ and $c_j>0$ for each $j=1,2,\dots,8$. Therefore, the result follows from
Proposition \ref{full-ranks}.
\end{proof}

In the statement (iii) of Theorem \ref{main}, some boundary points of $F$ have still full ranks as well as interior points.
In fact, every nine choice among ten
product vectors spans the whole space $\mathbb C^2\ot\mathbb
C^2\ot\mathbb C^2$. This is also the case when ten product vectors are
replaced by their partial conjugates. If we take eight product
vectors among ten, then they sometimes span the whole space, and
sometimes do not span the whole space. We show this for ten product vectors
$$
\begin{aligned}
\eta_1&=(1,+\omega^3)^\ttt \otimes (1,+\omega^1)^\ttt \otimes (u^{-1},+\omega^7)^\ttt,\\
\eta_2&=(1,+\omega^3)^\ttt \otimes (1,-\omega^1)^\ttt \otimes (u^{-1},-\omega^7)^\ttt,\\
\eta_3&=(1,-\omega^3)^\ttt \otimes (1,+\omega^1)^\ttt \otimes (u^{-1},-\omega^7)^\ttt,\\
\eta_4&=(1,-\omega^3)^\ttt \otimes (1,-\omega^1)^\ttt \otimes (u^{-1},+\omega^7)^\ttt,\\
\eta_j&=\bar\eta_{j-4},\qquad j=5,6,7,8,
\end{aligned}
$$
in (\ref{ten}) together with $|010\ran$ and $|101\ran$.

Motivated by the relation
$\lan (p,\alpha)^\ttt|(p^{-1},-\alpha)^\ttt\ran=0$ for $p>0$ and $\alpha\in\mathbb T$,
we define
$$
\begin{aligned}
\zeta_1&=(1,+\omega^3)^\ttt \otimes (1,-\omega^1)^\ttt \otimes (u,+\omega^7)^\ttt,\\
\zeta_2&=(1,+\omega^3)^\ttt \otimes (1,+\omega^1)^\ttt \otimes (u,-\omega^7)^\ttt,\\
\zeta_3&=(1,-\omega^3)^\ttt \otimes (1,-\omega^1)^\ttt \otimes (u,-\omega^7)^\ttt,\\
\zeta_4&=(1,-\omega^3)^\ttt \otimes (1,+\omega^1)^\ttt \otimes (u,+\omega^7)^\ttt,\\
\zeta_j&=\bar\zeta_{j-4},\qquad j=5,6,7,8.
\end{aligned}
$$
Then  the \lq coefficient matrix\rq\ $L:=[\lan \zeta_i|\eta_j\ran]_{i,j=1,2,\dots,8}$
of $\{|\eta_j\ran: j=1,2,\dots.8\}$ with respect to $\{|\zeta_i\ran\}$ is given by
\begin{equation}\label{coeff}
L=
\left(\begin{matrix}0&K\\  \bar K&0\end{matrix}\right)\in M_2(M_4),
\quad {\text{\rm with}}\
 K=2\sqrt 2 \omega^3
\left(\begin{matrix}
+&-&-&-\\
-&+&-&-\\
-&-&+&-\\
-&-&-&+
\end{matrix}\right),
\end{equation}
where $+$ and $-$ denote $+1$ and $-1$, respectively.
Because $\det K\neq 0$, we see that both $\{|\eta_j\ran\}$ and $\{|\zeta_j\ran\}$ are linearly independent.
If we take partial conjugates of $|\eta_j\ran$ and $|\zeta_i\ran$ in the $A$, $B$ and $C$ parties respectively,
then $K/2\sqrt 2$ is replaced by
$$
\omega^{-3}
\left(\begin{matrix}
-&+&-&-\\
+&-&-&-\\
-&-&-&+\\
-&-&+&-\end{matrix}\right),\quad
\omega^{-3} \left(\begin{matrix}
-&-&+&-\\
-&-&-&+\\
+&-&-&-\\
-&+&-&-\end{matrix}\right)\quad {\text{\rm and}}\quad
\omega^{-3}
\left(\begin{matrix}
-&-&-&+\\
-&-&+&-\\
-&+&-&-\\
+&-&-&-\end{matrix}\right),
$$
respectively. Therefore, all the partial conjugates of  $\{|\eta_j\ran\}$ are linearly independent. This gives another proof
of Proposition \ref{full-ranks}.

If one of $|\eta_j\ran$ is replaced by $|010\ran$ and $|101\ran$ respectively,  then one column of $L$ is replaced by
$$
u(\omega^3,\omega^7,\omega^3,\omega^7,\omega^5,\omega^1,\omega^5,\omega^1)^\ttt\quad {\text{\rm and}}\quad
(\omega^6,\omega^2,\omega^6,\omega^2,\omega^2,\omega^6,\omega^2,\omega^6)^\ttt,
$$
respectively. They are still non-singular, and the eight product vectors we have chosen are linearly independent.
Especially, we see that every nine choice among ten product vectors spans the whole space $\mathbb C^2\ot\mathbb
C^2\ot\mathbb C^2$. If all the  $|\eta_j\ran$'s and $|\zeta_i\ran$'s are replaced by partial conjugate,
then we have the same conclusion.
Therefore, we have the following:

\begin{proposition}\label{nine}
Every nine choice among ten product vectors in
$$
P_W\cap P_{W_A}\cap P_{W_B},\quad
P_W\cap P_{W_B}\cap P_{W_C},\quad
P_W\cap P_{W_C}\cap P_{W_A}\quad {\text{\rm or}}\quad
P_{W_A}\cap P_{W_B}\cap P_{W_C}
$$
spans the whole space $\mathbb C^2\ot\mathbb C^2\ot\mathbb C^2$.
\end{proposition}

We give a geometric interpretation of Proposition \ref{nine}. We begin with the real vector space of all
self-adjoint three qubit matrices, and the $10$-dimensional subspace
generated by ten pure product states in Theorem \ref{main}. We consider the $9$-dimensional affine manifold $H$
given by the condition of trace one. Then the dual face $F$ in Theorem \ref{main}
coincides with the convex set ${\mathcal S}_H={\mathcal S}\cap H$ of all separable states on the affine manifold $H$.
The boundary of ${\mathcal S}_H$
consists of maximal faces isomorphic to the $8$-dimensional simplex.
We also consider the set ${\mathcal T}_H={\mathcal T}\cap H$ of all PPT states on $H$.
Proposition \ref{nine} tells us that
the interior of these maximal faces
are contained in the interior of ${\mathcal T}_H$. Take an interior point $\varrho_0$
of ${\mathcal S}_H$ and an interior point $\varrho_1$ of a maximal face of ${\mathcal S}_H$. Take also the line
segment $\varrho_t=(1-t)\varrho_0+t\varrho_1$ from $\varrho_0$ to $\varrho_1$.
Because $\varrho_1$ is an interior point of ${\mathcal T}_H$, there exists $t>1$ so that $\varrho_t\in {\mathcal T}_H$,
which is a PPT entangled state. In this way, we have bunch of PPT entanglement, as in
\cite{{ha-kye-sep-face},{ha-kye-multi-unique}}.
Note that $\varrho_1$ is a separable state of length $9$ with full ranks and unique decomposition.

The same argument may hold even if we take $\varrho_1$ on the boundary of a maximal face.
In fact, if we take the face generated by eight pure product states corresponding to $|\eta_j\ran$
with $j=1,2,\dots,8$, then this face is on the boundary of a maximal face, but still in the interior
of ${\mathcal T}_H$ by Proposition \ref{full-ranks}. In this case, $\varrho_1$ is a separable state
of length $8$ with full ranks and unique decomposition.
But, this is not the case in general. Suppose that two of $|\eta_1\ran, \dots,|\eta_4\ran$ are replaced by
$|010\ran$, $|101\ran$ to get the eight product vectors. Then we see that the matrix $\bar K$ should be replaced by a matrix
with two columns of same direction, while the lower-right corner is still zero. This implies that
the coefficient matrix with respect to $\{|\zeta_j\ran\}$ is singular, and those eight product vectors are
linearly dependent. The convex hull of the corresponding eight pure product states generate
a face isomorphic to the $7$-dimensional simplex which is contained in the boundary of ${\mathcal T}_H$,
and so the line segment from $\varrho_0$ to an interior point of this face cannot be extended
in ${\mathcal T}_H$. If we replace one of  $|\eta_1\ran, \dots,|\eta_4\ran$ by $|010\ran$ and
replace one of $|\eta_5\ran, \dots,|\eta_8\ran$ by $|101\ran$ then we can see that the resulting coefficient matrix
is non-singular.


\end{document}